\documentclass[12pt,a4paper]{article}

\usepackage{amssymb,amsmath,amsthm}
\usepackage[english]{babel}
\usepackage{t1enc}
\usepackage[latin2]{inputenc}
\usepackage{epsfig}

\newtheorem*{thm}{Theorem}
\newtheorem*{cor}{Corollary}
\newtheorem*{defi}{Definition}

\newtheorem*{prop}{Proposition}

\usepackage{indentfirst}
\newcommand{\NPc}{{\sc NP}-complete }
\newcommand{\NPcx}{{\sc NP}-complete}
\newcommand{\NAES}{{\sc NAE-SAT }}
\newcommand{\NAESx}{{\sc NAE-SAT}}

\begin{document}

\title{Partitionability to two trees is {\sc NP}-complete}
\author{D\"om\"ot\"or P\'alv\"olgyi
\thanks{Dept of Comp. Sci. and Communication Networks Lab, ELTE, Budapest.\newline
Supported by OTKA T67867.}
}
\date{}
\maketitle

\begin{abstract}
\noindent 

We show that {\sc P2T} - the problem of deciding whether the edge set of a simple graph can be partitioned into two trees or not - is \NPcx.

\end{abstract}

\medskip
\noindent

It is a well known that deciding whether the edge set of a graph can be partitioned into $k$ spanning trees or not is in $P$ \cite{E}. Recently Andr\'as Frank asked what we know about partitioning the edge set of a graph into $k$ (not necessarily spanning) trees. One can easily see that whether a simple graph is a tree or not is in {\sc P}. It was shown by Kir\'aly that the problem of deciding if the edge set of a simple graph is the disjoint union of three trees is \NPc by reducing the {\em 3-colorability} problem to it \cite{K}. Now we prove that for two trees the problem is also \NPcx. First we define the problem precisely.

\begin{defi} The input of the decision problem {\sc P2T} is a graph $G=(V,E)$ and the goal is to decide whether there is an $E=E_1\dot\cup E_2$ partitioning of the edge set such that both $E_1$ and $E_2$ form a tree.
\end{defi}
 
\begin{thm} {\sc P2T} is \NPcx.
\end{thm}

It is obvious that {\sc P2T} belongs to {\sc NP}. To prove its completeness, we will show that the \NAES ({\sc Not-All-Equal SAT }problem) is reducible to {\sc P2T}.\\

The \NAES problem is the following. We are given polynomially many clauses over the variables $x_1,\ldots,x_n$ and we have to decide whether there is an evaluation of the variables such that each clause contains both a true and a false literal. This is called a {\em good} evaluation. Eg., if the formula contains at least one clause of size one (like $x_1$), it does $not$ have a {\em good} evaluation. This problem is well known to be \NPc \cite{GJ}.\\

Now we will construct a graph $G$ from a given clause set $\mathcal C$. We denote its variables by $x_1, \ldots, x_n$. The graph $G$ will consist of two main parts, $L_i$ and $C_j$ type subgraphs. A subgraph $L_i$ corresponds to each variable, while a subgraph $C_j$ corresponds to each clause from $\mathcal C$. Beside the $L_i$'s corresponding to the variables, we also have two extra subgraphs of this type, $L_0$ and $L_{n+1}$. The vertex sets corresponding to the clauses and variables are all disjoint, except for $V(L_i) \cap V(L_{i+1})$, what is a single vertex denoted by $t_i$.\\

A subgraph $L_i$ corresponding to the variable $x_i$ consists of four vertices that form a cycle in the following order: $t_{i-1}, v_i, t_i$ and $\overline{v}_i$. We would like to achieve that one of the trees contains the edges from $t_{i-1}$ through $v_i$ to $t_i$, while the other from $t_{i-1}$ through $\overline{v}_i$ to $t_i$. For simplicity, we denote $t_{-1}$ by $\alpha$ and $t_{n+1}$ by $\omega$. Both trees will have to contain a path from $t_0$ to $t_n$
.  The idea is that we want to force one of the trees to go through exactly those $v_i$'s for which $x_i$ is true.\\

Before we start the construction of the subgraphs corresponding to the clauses, we introduce a notation. We say that two vertices $u$ and $w$ are linked with a {\em purple} edge if\\
(1) There is no edge between $u$ and $w$.\\
(2) The smaller connectivity component of $G\setminus\{u,w\}$ (called {\em purple subgraph}) consists of four vertices: $v^{uw}_1, v^{uw}_2, v^{uw}_3$ and $v^{uw}_4$.\\
(3) The $v^{uw}_i$ vertices form a cycle in this order.\\
(4) The $v^{uw}_i$ vertices are not connected to any other vertices, except for $v^{uw}_1$ that is connected to $u$ and $w$. ({\em See Figure 1.})\\
This is a very useful structure because if $E(G)$ is the union of two trees, then they both have to enter this purple subgraph since a tree cannot contain a cycle. So if the vertices are linked with a purple edge and $E(G)=E(T)\dot\cup E(F)$ (where $T$ and $F$ denote the two trees), then it means that $uv^{uw}_1\in E(T)$ and $wv^{uw}_1\in E(F)$ or $uv^{uw}_1\in E(F)$ and $wv^{uw}_1\in E(T)$.

\vskip 0.2cm

\epsfxsize=4truecm
\centerline{\epsffile{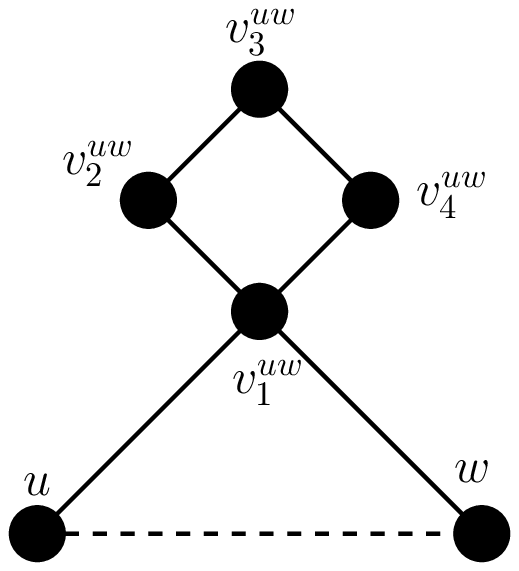}}

\vskip 0.2cm

\centerline{{\bf Figure 1.} A $purple$ edge.}

\vskip 0.4cm

A subgraph $C_j$ corresponding to the $j$th clause consists of $3k$ vertices where $k$ is the size of the $j$th clause whose literals are denoted by $l^j_1, \ldots, l^j_k$. A cycle of length $2k$ is formed by the following vertices in this order: $p^j_1,$ $q^j_1,$ $p^j_2,$ $q^j_2, \ldots, p^j_k,$ $q^j_k$. The other $k$ vertices are denoted by $r^j_1,\ldots, r^j_k$. The vertex $r^j_i$ is always connected to $p^j_i$ and it is also connected to $v_m$ if $l^j_i$ is $x_m$ or to $\overline v_m$ if $l^j_i$ is $\overline x_m$. Furthermore, there is a purple edge between $r^j_i$ and $q^j_i$. This will ensure that a tree ``entering'' the clause subgraph through an $r^j_i$, cannot ``leave'' this subgraph. ({\em See Figure 2. for a graph with one clause. The dashed edges mean purple edges.})\\

\vskip 0.2cm

\epsfxsize=7.6truecm
\centerline{\epsffile{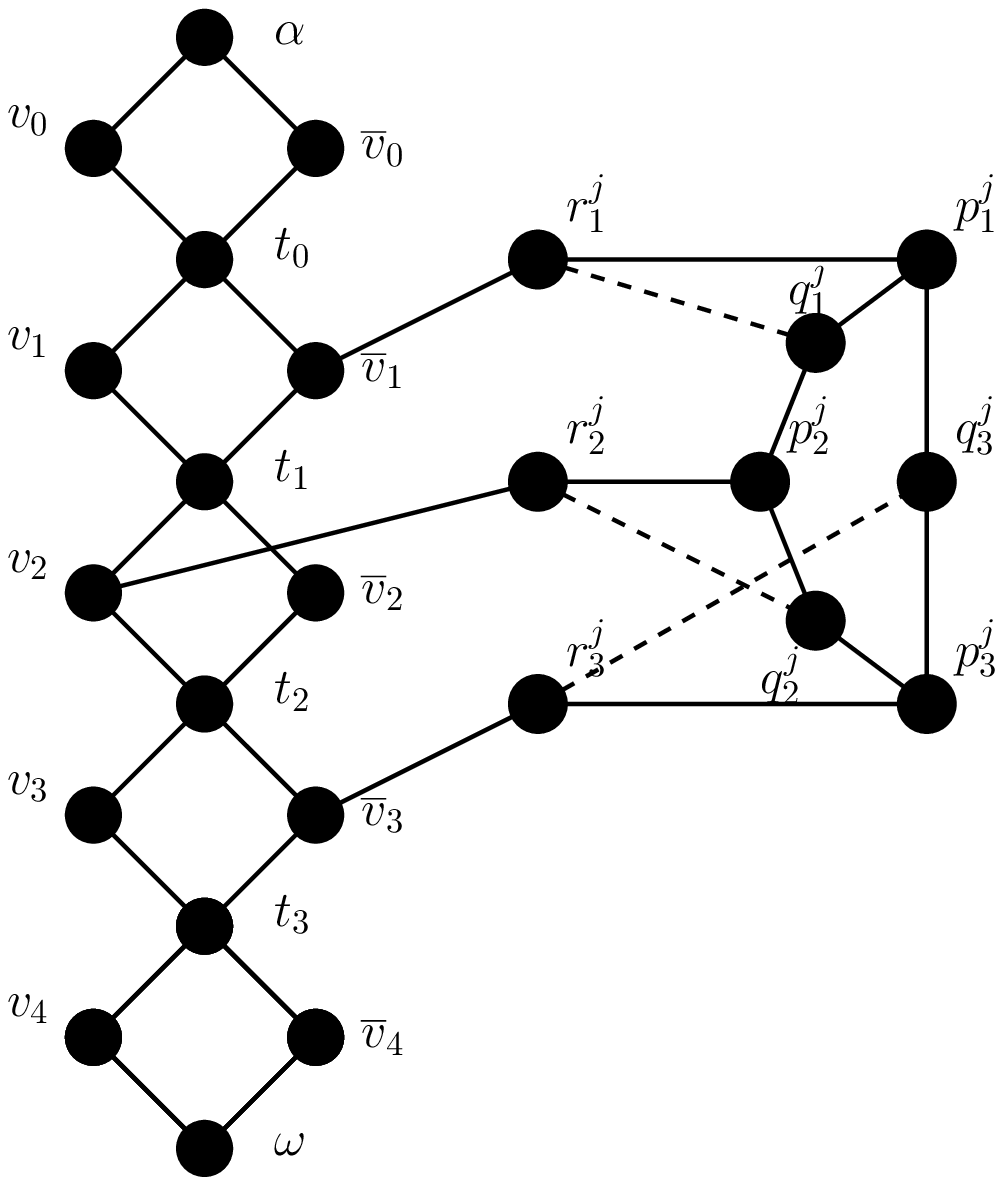}}

\vskip 0.2cm

\centerline{{\bf Figure 2.} The graph for the single clause $\overline x_1 \vee x_2 \vee \overline x_3$.}

\vskip 0.4cm

The construction is finished, now we have to prove it's correctness. The easier part is to show that if our \NAES problem has a good evaluation, then we can partition the edges into two trees, $T$ and $F$. First let us fix a good evaluation. Let the tree $T$ contain the path from $\alpha$ to $\omega$ through the $v_i$'s for true $x_i$'s and through the $\overline v_i$'s for false $x_i$'s (consider the non-existing $x_0$ and $x_{n+1}$ true). Similarly $F$ trails from $\alpha$ to $\omega$ through the $\overline v_i$'s for true $x_i$'s and through the $v_i$'s for false $x_i$'s. So each $v_i$ and each $\overline v_i$ belongs to exactly one of the trees. If a tree contains $v_i$ (or $\overline v_i$), let it also contain the $v_ir^j_m$ (or $\overline v_ir^j_m$) and $r^j_mp^j_m$ edges if $x_i$ (or $\overline x_i$) is in the $j$th clause. This way both trees enter each clause subgraph since the evaluation satisfied our \NAES problem. Let the edges $p^j_iq^j_i$ and $q^j_ip^j_{i+1}$ belong to the tree that does $not$ contain $r^j_i$. This guarantees that we can enter the purple subgraph belonging to $r^j_i$ and $q^j_i$ by both trees. It can be also easily seen that the edges of $T$ (and of $F$) form a tree and every edge is assigned to one of them. So we are done with this part.\\

To prove the other part, let us suppose that $E(G)=T\dot\cup F$ for two trees $T$ and $F$. We know that $t_0 \in V(T)$ and $t_0 \in V(F)$, because the $L_0$ subgraph is a cycle. We can suppose that $v_0 \in V(T)$, $\overline v_0 \in V(F)$, $\alpha \in V(T)$ and $\alpha \in V(F)$. We can similarly suppose $\omega\in V(T)$ and $\omega\in V(F)$. Let us direct all the edges of the trees away from $\alpha$.

\begin{prop} There are no edges coming out of the purple subgraphs.
\end{prop}
\begin{proof} Both trees have to enter each purple subgraph since a tree cannot contain a cycle and since there are only two edges connecting a purple subgraph to the rest of the graph, both of them must be directed toward the purple subgraph.
\end{proof}

We may conclude that the trees cannot ``go through'' purple edges.

\begin{prop} There are no edges coming out of the clause subgraphs.
\end{prop}
\begin{proof} Let us suppose that the edge from $r^j_i$ going to some $v_m$ (or $\overline v_m$) is directed away from $r^j_i$ and is in $T$. This implies $p^j_ir^j_i \in T$ as well because $T$ cannot enter $r^j_i$ through the purple edge. But because $r^j_i\notin V(F)$, therefore $q^j_i \in V(F)$ since they are linked with a purple edge. This shows $q^j_ip^j_i \notin T$, so $T$ must have entered $p^j_i$ from $p^j_{i-1}$ through $q^j_{i-1}$. But then $q^j_{i-1}\notin V(F)$, so $r^j_{i-1}\in V(F)$. This means $T$ entered $p^j_{i-1}$ from $p^j_{i-2}$. And we can go on so until we get back to $p^j_i$, what gives a contradiction.
\end{proof}

So now we know that the clauses are dead ends as well as the purple subgraphs. Since $T$ and $F$ trail from $\alpha$ to $\omega$, each $v_i$ (and $\overline v_i$) must be contained in exactly one of them. So we can define $x_i$ to be true if and only if $v_i \in V(T)$. Now the only property left to show is that the literals in the clauses are not equal. But if they were equal in the $j$th clause, then the $C_j$ subgraph corresponding to this clause would be entered by only one of the trees and hence that tree would contain a cycle, contradiction. So we have shown that each tree partition yields a proper evaluation. This finishes the proof of the theorem. $\Box$\\

Now we prove an upper bound on the maximum degree of the graph that we constructed. The degree of every vertex, except the $v_i$'s and $\overline v_i$'s, is at most four. A $v_i$ (or $\overline v_i$) has degree equal to two plus the number of occurrences of $x_i$ (or $\overline x_i$) in the clause set. But a \NAES problem is easily reducible to a \NAESx-(2;2) problem (meaning that each literal can occur at most twice). If a literal $l$ would occur in at least three clauses, then let us execute the following operation until we have at most two of each literal. Replace $(C_1' \vee l), (C_2' \vee l), (C_3' \vee l)$ with $(l \vee \overline z), (C_1' \vee l), (C_2' \vee z), (C_3' \vee z)$ where $z$ is a new variable.

\begin{cor} The decision of whether the edge set of a simple graph is the disjoint union of two trees or not, is \NPc even for graphs with maximum degree four.
\end{cor}

\noindent
{\bf Acknowledgment.} I would like to thank Zoltán Király for discussions and the anonymous referee for his useful comments.

\end{document}